\documentclass[11pt]{article}
\usepackage{fullpage}
\usepackage{caption}

\usepackage{mdframed}

\def\showdraftbox{0}
\def\showcolorlinks{1}

\usepackage{etex}

\usepackage[l2tabu, orthodox]{nag}

\usepackage[utf8]{inputenc}

\usepackage{xspace,enumerate}

\usepackage[dvipsnames]{xcolor}

\usepackage[T1]{fontenc}
\usepackage[full]{textcomp}

\usepackage[american]{babel}

\usepackage{mathtools}

\usepackage{bm}

\usepackage{amsthm}

\newtheorem{theorem}{Theorem}[section]
\newtheorem*{theorem*}{Theorem}

\newtheorem{proposition}[theorem]{Proposition}
\newtheorem*{proposition*}{Proposition}
\newtheorem{lemma}[theorem]{Lemma}
\newtheorem*{lemma*}{Lemma}

\newtheorem*{conjecture*}{Conjecture}

\newtheorem*{fact*}{Fact}

\newtheorem*{hypothesis*}{Hypothesis}

\theoremstyle{definition}
\newtheorem{definition}[theorem]{Definition}

\newtheorem{question}[theorem]{Question}

\theoremstyle{remark}
\newtheorem{claim}[theorem]{Claim}
\newtheorem*{claim*}{Claim}
\newtheorem{remark}[theorem]{Remark}
\newtheorem*{remark*}{Remark}

\newtheorem*{observation*}{Observation}

\ifnum\showcolorlinks=1
\usepackage[
pagebackref,
letterpaper=true,
colorlinks=true,
urlcolor=blue,
linkcolor=blue,
citecolor=OliveGreen,
]{hyperref}
\fi

\ifnum\showcolorlinks=0
\usepackage[
pagebackref,
letterpaper=true,
colorlinks=false,
pdfborder={0 0 0}
]{hyperref}
\fi

\usepackage{prettyref}

\newcommand{\savehyperref}[2]{\texorpdfstring{\hyperref[#1]{#2}}{#2}}

\newrefformat{eq}{\savehyperref{#1}{\textup{(\ref*{#1})}}}
\newrefformat{lem}{\savehyperref{#1}{Lemma~\ref*{#1}}}
\newrefformat{def}{\savehyperref{#1}{Definition~\ref*{#1}}}
\newrefformat{thm}{\savehyperref{#1}{Theorem~\ref*{#1}}}
\newrefformat{cor}{\savehyperref{#1}{Corollary~\ref*{#1}}}
\newrefformat{cha}{\savehyperref{#1}{Chapter~\ref*{#1}}}
\newrefformat{sec}{\savehyperref{#1}{Section~\ref*{#1}}}
\newrefformat{app}{\savehyperref{#1}{Appendix~\ref*{#1}}}
\newrefformat{tab}{\savehyperref{#1}{Table~\ref*{#1}}}
\newrefformat{fig}{\savehyperref{#1}{Figure~\ref*{#1}}}
\newrefformat{hyp}{\savehyperref{#1}{Hypothesis~\ref*{#1}}}
\newrefformat{alg}{\savehyperref{#1}{Algorithm~\ref*{#1}}}
\newrefformat{rem}{\savehyperref{#1}{Remark~\ref*{#1}}}
\newrefformat{item}{\savehyperref{#1}{Item~\ref*{#1}}}
\newrefformat{step}{\savehyperref{#1}{step~\ref*{#1}}}
\newrefformat{conj}{\savehyperref{#1}{Conjecture~\ref*{#1}}}
\newrefformat{fact}{\savehyperref{#1}{Fact~\ref*{#1}}}
\newrefformat{prop}{\savehyperref{#1}{Proposition~\ref*{#1}}}
\newrefformat{prob}{\savehyperref{#1}{Problem~\ref*{#1}}}
\newrefformat{claim}{\savehyperref{#1}{Claim~\ref*{#1}}}
\newrefformat{relax}{\savehyperref{#1}{Relaxation~\ref*{#1}}}
\newrefformat{red}{\savehyperref{#1}{Reduction~\ref*{#1}}}
\newrefformat{part}{\savehyperref{#1}{Part~\ref*{#1}}}

\newcommand{\problemmacro}[1]{\texorpdfstring{\textsc{#1}}{#1}\xspace}

\ifnum\showdraftbox=1
\newcommand{\draftbox}{\begin{center}
  \fbox{%
    \begin{minipage}{2in}%
      \begin{center}%
          \Large\textsc{Working Draft}\\%
        Please do not distribute%
      \end{center}%
    \end{minipage}%
  }%
\end{center}
\vspace{0.2cm}}
\else
\newcommand{\draftbox}{}
\fi

\title{Finding Perfect Matchings in Bipartite Hypergraphs}

\author{Chidambaram Annamalai\thanks{Department of Computer Science,
    ETH Zurich. Email:
\href{mailto:cannamalai@inf.ethz.ch}{cannamalai@inf.ethz.ch}. \newline Work performed while the author was at the School of Basic Sciences, EPFL.}}

\date{\today}

\begin{document}

\maketitle
\draftbox
\thispagestyle{empty}

\begin{abstract}
  Haxell's condition~\cite{haxell1995condition} is a natural \emph{hypergraph} analog of Hall's condition, which is a well-known necessary and sufficient condition for a bipartite graph to admit a perfect matching. That is, when Haxell's condition holds it forces the existence of a perfect matching in the bipartite hypergraph. Unlike in graphs, however, there is no known polynomial time algorithm to find the hypergraph perfect matching that is guaranteed to exist when Haxell's condition is satisfied.

We prove the existence of an efficient algorithm to find perfect matchings in bipartite hypergraphs whenever a stronger version of Haxell's condition holds. Our algorithm can be seen as a generalization of the classical \emph{Hungarian algorithm} for finding perfect matchings in bipartite graphs. The techniques we use to achieve this result could be of use more generally in other combinatorial problems on hypergraphs where disjointness structure is crucial, e.g. \problemmacro{Set Packing}.


\end{abstract}

\medskip
\noindent
{\small \textbf{Keywords:}
bipartite hypergraphs, matchings, local search algorithms.
}

\newpage
\section{Introduction}
Bipartite matchings are a ubiquitous quantity across science and
engineering. The task of finding a maximum matching (or, specifically,
a perfect matching) in a bipartite graph captures a fundamental notion
of assignment that has turned out to have wide applicability. One
reason their influence has been felt so deeply is essentially
\emph{computational}. The basic fact that maximum matchings in
bipartite graphs can be found efficiently is arguably the most
important factor in determining their widespread use. Confirming the
importance of this problem, decades of research in theoretical
computer science has contributed to increasingly faster and more
sophisticated algorithms for finding maximum matchings in bipartite
graphs. These include connections to other fundamental problems like
matrix
multiplication~\cite{lovasz1979determinants,mucha2004maximum}. Most
recently, the exciting work of M\c{a}dry~\cite{madry2013navigating} breaks
the decades-old $O(m\sqrt{n})$
Hopcroft-Karp-Karzanov~\cite{hopcroft1973n,karzanov1973nakhozhdenii}
barrier for finding maximum matchings in bipartite graphs.

In this paper we look at the analogous problem in the hypergraph
setting. We address the following question: Do there exist
efficient---in the sense of \emph{polynomial} running time---algorithms
to find perfect matchings in bipartite \emph{hypergraphs}?

In an $r$-uniform bipartite hypergraph the vertex set is partitioned
into two sets $A$ and $B$ such that each edge contains exactly one
vertex from $A$ and $r-1$ vertices from $B$. A perfect matching is a
collection of disjoint edges such that each vertex in $A$ is covered
by exactly one edge in the collection.

In order for the question to make any sense at all we need to impose
some additional restrictions on the input as any such algorithm that
works unconditionally even for the case $r=3$ is tantamount to P$=$NP,
as the trivial reduction from \problemmacro{$3$-Dimensional
  Matching}\footnote{For a definition of this NP-complete problem see, for example,~\cite{cygan2013improved}.} shows. In this sense it is not surprising that the question of
\emph{finding} perfect matchings in bipartite hypergraphs was not
considered before. However, as we will see, under certain conditions
the problem becomes interesting algorithmically. The starting point
for such investigations is to ask ourselves if there is a condition
similar to Hall's~\cite{hall1935representatives} condition that
guarantees the existence of perfect matchings in bipartite
hypergraphs. This question was solved in a very satisfying way by
Haxell~\cite{haxell1995condition} in the mid-90s leading to a striking
generalization of Hall's theorem. The condition is that for every
subset $S$ of $A$, the size of the hitting set of hyperedges incident
to $S$ must be proportional to the size of $S$. This condition
is sufficient to force the existence of a perfect matching in the
hypergraph. More formally, given a bipartite hypergraph $H=(A,B,E)$,
for a set $S \subseteq A$ let
$E_S := \{ e\in E \;|\; |e \cap S| = 1\} $ be the set of hyperedges of
$H$ incident to $S$. For a given collection of edges $F \subseteq E$,
define $\tau(F)$ to be the smallest cardinality subset of $B$
that hits\footnote{A subset $S \subseteq B$ \emph{hits} all the edges
  in $F$ if for each $e \in F$, $S \cap e \neq \emptyset$.} all the
edges in $F$.

\begin{theorem}[Haxell~\cite{haxell1995condition}]\label{thm:haxell}
  Let $H = (A,B,E)$ be an $r$-uniform bipartite hypergraph. If
  \[\tau(E_S) > (2r-3)(|S|-1) \quad \forall S \subseteq A,\] then $H$ admits a perfect matching.
\end{theorem}

There are several interesting aspects to
Theorem~\ref{thm:haxell}. First, it reduces to Hall's condition when
$r=2$. Second, the statement is ``tight'' in the sense that it is
\emph{not} true when the strict inequality is replaced by a non-strict
one, i.e., for every $r$ there is an $r$-uniform bipartite hypergraph
that satisfies Haxell's condition with non-strict inequalities and yet
contains no perfect matching. Finally, in constrast to the graph case,
the proof is not constructive and does not lead to an efficient
algorithm that also finds the perfect matching.

Besides being interesting objects in their own right, perfect
matchings in bipartite hypergraphs have become a crucial concept in
recent
work~\cite{bansal2006santa,feige2008allocations,haeupler2011new,asadpour2012santa,svensson2012santa,polacek2012quasi,DBLP:conf/soda/AnnamalaiKS15}
on a particular allocation problem, called \problemmacro{Restricted
  Max-Min Fair Allocation}, where the goal is to partition
\emph{indivisible} resources among players in a balanced manner. The
latest work that exploits this
connection~\cite{DBLP:conf/soda/AnnamalaiKS15} has further shown that
local search algorithms based on alternating trees for this problem
can be made to run in polynomial time. These developments brought to
light the question posed towards the beginning of this section. Is it
possible to efficiently find perfect matchings in bipartite
hypergraphs by assuming a stronger version of Haxell's condition?

The question was also raised implicitly by Haeupler, Saha and
Srinivasan~\cite{haeupler2011new} who write that ``Haxell's theorems
are again highly non-constructive \ldots.'' As a direction for future
research it was therefore asked in~\cite{DBLP:conf/soda/AnnamalaiKS15} if
ensuring the much stronger condition
\begin{equation}\label{eqn:stronger}
\tau(E_S) \geq 100r|S| \quad \forall S \subseteq A
\end{equation}
was sufficient to also find the perfect matching efficiently. The
large constant before $r$ reflected the belief that some loss is to be
expected, at least by using their techniques. The allocation algorithm
underlying their work turned out to yield a weaker guarantee under
strengthenings similar to \eqref{eqn:stronger}. In particular
their techniques yield the following theorem.

\begin{theorem}[\cite{DBLP:conf/soda/AnnamalaiKS15}]\label{thm:weak}
  Let $c > 0$ be some absolute constant. Choose any
  $0 < \epsilon \leq 1$, $r \geq 2$ and consider $r$-uniform bipartite
  hypergraphs $H=(A,B,E)$ that satisfy
  $\tau(E_S) \geq (c/\epsilon) \cdot r(|S|-1).$ For such a family of
  hypergraphs there is a polynomial time algorithm that assigns one
  hyperedge $e_a \in E$ for every vertex $a \in A$ such that it is
  possible to choose disjoint subsets
  $\{S_a \subseteq e_a \cap B\}_{a\in A}$ of cardinality at least
  $(1-\epsilon)(r-1)$.
\end{theorem}

Whereas in the allocation setting sufficiently large subsets
$\{S_a\}_{a \in A}$ also correspond to a good allocation, in a
hypergraph setting Theorem~\ref{thm:weak} does not guarantee a
collection of valid and disjoint hyperedges unless $\epsilon < 1/r$.
However, for a choice of $\epsilon$ in latter range, the assumption
$\tau(E_S) \geq (c/\epsilon) \cdot r(|S|-1)$ is much stronger than the
one in Theorem~\ref{thm:haxell}.

\paragraph{Our results} Our main result is that a suitable
constructivization of Theorem~\ref{thm:haxell} is indeed possible. We
prove the following.

\begin{theorem}\label{thm:main}
  For every fixed choice of $\epsilon >0$ and $r \geq 2$, there exists
  an algorithm $\mathcal{A}(\epsilon, r)$ that finds, in time
  polynomial in the size of the input, a perfect matching in
  $r$-uniform biparite hypergraphs $H=(A,B,E)$ satisfying
  \[\tau(E_S) > (2r-3 + \epsilon)(|S|-1) \quad \forall S \subseteq
  A.\]
\end{theorem}


Notice that such an algorithm with a polynomial running time
dependence on $1/\epsilon$ would be able to efficiently find perfect
matchings in bipartite graphs only assuming Haxell's
condition. Currently we see no way of achieving such a result. In
particular, our techniques make essential use of the $\epsilon$
strengthening of Haxell's condition as assumed in
Theorem~\ref{thm:main}. We emphasize that the running time of
$\mathcal{A}$ depends exponentially on $r$ and $1/\epsilon$, which
explains the particular order of the quantifiers in the statement. See
also Theorem~\ref{thm:maincorollary} for a slightly stronger corollary of
our main result.

An outline of the ideas behind Theorem~\ref{thm:main} requires setting
up some context involving previous work, which we do presently.

\paragraph{Context} It helps to start with the graph case. Here the
basic augmenting algorithm (often called the \emph{Hungarian
  algorithm} after K\H onig and
Egerv\'ary~\cite{west2001introduction}) takes a partial matching and
constructs an alternating tree of unmatched and matched edges with an
unmatched vertex at the root. For convenience we can imagine this tree
partitioned into ``layers'', where the $i$th layer contains all vertices
at distance $2i-1$ and $2i$ from the root (along with their associated
edges in the tree). Matched edges appearing in the tree can also be
called ``blocking'' since they prevent us from augmenting the partial
matching immediately. If a leaf of the alternating tree happens to be
an unmatched vertex (i.e., the leaf edge is not blocked by some edge
in the partial matching) then the corresponding root to leaf path is
an augmenting path for the considered matching and thus the augmenting
algorithm terminates. The fact that such a leaf always exists is
guaranteed by Hall's condition~\cite{hall1935representatives}. The
proof of Haxell's theorem (Theorem~\ref{thm:haxell}) involves a
similar alternating tree, the key difference being that a single
hyperedge in the alternating tree may now be blocked by \emph{several}
hyperedges (up to $r-1$) from the partial matching. Therefore, even if
a leaf hyperedge in the tree does not intersect any hyperedges from
the partial matching, we may not be able to immediately augment the
partial matching like in the graph case. What we can do is only swap
the corresponding blocking edge (the unique ancestor of the unblocked
leaf edge) with the leaf edge in the partial matching and
continue. When there are no longer leaf edges in the alternating tree,
the existence of a vertex disjoint hyperedge for one of the $A$
vertices in the tree is implied by Haxell's condition. This
alternating tree algorithm for hypergraph matchings by
Haxell~\cite{haxell1995condition}, which underlies
Theorem~\ref{thm:haxell}, is not known to make fewer than
exponentially many modifications (swapping operations) to the partial
matching before termination (at which point the root is matched).

To make such a local search algorithm
efficient~\cite{DBLP:conf/soda/AnnamalaiKS15} devised a similar but
different algorithm that ensures a (constant factor) multiplicative
increase in the number of blocking edges from layer to layer. This
guarantees that the height of the alternating tree is always
logarithmic, leading to ``short'' augmenting paths. They also avoid
making changes to the partial matching unless sufficiently many
changes can be made at once. In other words, the partial matching is
updated lazily. Coupled with other ideas, this leads to a polynomial
time combinatorial allocation algorithm that achieves their main
result. For the hypergraph setting, however, this only yields
Theorem~\ref{thm:weak}.

\paragraph{Our Techniques} One obstacle with the algorithm
of~\cite{DBLP:conf/soda/AnnamalaiKS15}, is that a single blocking edge
can block up to $r-1$ hyperedges in the same layer. This effect can
accumulate across consecutive layers preventing the desired growth in
the number of blocking edges across the layers of the alternating
tree, which we require in order to guarantee a logarithmic bound on
the height of the alternating tree. This makes it important to view
the structure of the blocking edges when the layers of the alternating
tree are constructed. On the other hand the problem with the regular
alternating tree algorithm for hypergraph
matchings~\cite{haxell1995condition} is that a single $A$ vertex in a
layer can be part of an unbounded number of hyperedges in the next
layer in the alternating tree. This skews any subsequent progress made
by the alternating tree algorithm vastly in favor of a few $A$
vertices in the previous layers. To avoid this we impose a degree
bound on the $A$ vertices in the alternating tree, making the progress
more balanced among $A$ vertices in the same layer. In
Section~\ref{section:analysis} we show, despite imposing this upper
bound, a multiplicative growth in the number of blocking edges from
layer to layer. Next, since the structure of the blocking edges was
considered when constructing a layer, this creates complications when
we modify the partial matching and some layer in the tree. For
example, when some blocking edges are removed by swapping operations
in a layer it may be possible to have additional hyperedges for some
of the $A$ vertices in the same layer. At this point our algorithm
performs a so-called ``superposed-build'' operation (see
Section~\ref{section:augmenting2}) on the layer to check if
sufficiently many new hyperedges can be included. If so, it commits
the changes, otherwise it ignores the newly available hyperedges.

\subsection{Related work}
Most relevant to the result of this paper is the line of work
concerning \emph{alternating tree} algorithms for hypergraph matchings
starting with the work of Haxell~\cite{haxell1995condition}. The
algorithmic question of whether the underlying local search algorithm
can be made efficient was not considered until the work of Asadpour,
Feige and Saberi~\cite{asadpour2012santa}. They uncovered a beautiful
connection to strong integrality gaps for configuration linear
programs for allocation problems. This direction was subsequently also
pursued by Svensson~\cite{svensson2012santa} leading to a
breakthrough in the context of scheduling. Both
results~\cite{asadpour2012santa,svensson2012santa} were
non-constructive and only proved integrality gap upper
bounds. Following these results it became an important question if
such approaches based on alternating trees can be turned into
efficient algorithms with similar guarantees. Pol\'a\v{c}ek and
Svensson~\cite{polacek2012quasi} obtained significant savings leading
to a quasipolynomial time alternating tree algorithm for
\problemmacro{Restricted Max-Min Fair Allocation}, but it is still not
clear if their approach can be made truly polynomial. Building on
these ideas, a polynomial time alternating tree algorithm for the same
problem was obtained by Annamalai, Kalaitzis and
Svensson~\cite{DBLP:conf/soda/AnnamalaiKS15}.

The success of local search for combinatorial problems on hypergraphs
where disjointness structure is crucial has been a recurring theme in
the literature on \problemmacro{$k$-Set
  Packing}~\cite{karp1972reducibility}. Hurkens and
Schrijver~\cite{hurkens1989size} showed a $(k/2+\epsilon)$
approximation algorithm using an intuitive local search
algorithm. Halld\'orsson~\cite{halldorsson1995approximating} then
obtained a quasipolynomial $(k+2)/3$-approximation. Using a different
approach this was improved by Cygan, Grandoni, and
Mastrolilli~\cite{cygan2013sell} to a quasipolynomial time
$(k+1+\epsilon)/3$-approximation. A polynomial time
$(k+2)/3$-approximation was obtained by Sviridenko and
Ward~\cite{sviridenko2013large} using color-coding techniques. The
best known result for \problemmacro{$k$-Set Packing} to date is a
$(k+1+\epsilon)/3$-approximation due to
Cygan~\cite{cygan2013improved}, and also by Furer and
Yu~\cite{furer2014approximating}. It is interesting to note that all
of these results are based on local search. We believe our techniques
to be a useful addition to this repertoire.

In an important direction of research Chan and
Lau~\cite{chan2012linear} consider the power of linear and
semidefinite relaxations for the \problemmacro{$k$-Set Packing
  problem}. They show that a particular LP relaxation has integrality
gap at most $(k+1)/2$. A different LP relaxation arrived at by
applying $O(k^2)$ rounds of Chv\'atal-Gomory cuts to the standard LP
relaxation was also shown to have no worse integrality gap by Singh
and Talwar~\cite{singh2010improving}. It remains interesting to
consider the applicability of ``alternating tree'' style analyses, as
presented in this paper, to better understand the integrality gaps of
such strong LP and SDP relaxations.

For a different notion of bipartiteness in hypergraphs, Conforti
et al.~\cite{conforti1996perfect} study sufficient conditions for the
existence of perfect matchings. We also mention that for the case of
general hypergraphs, sufficient conditions in the spirit of Dirac's
theorem for graphs~\cite{dirac1952some} are known (see Alon
et al.~\cite{alon2012large} and references therein).

\section{Preliminaries}\label{section:preliminaries}
\begin{definition}[Bipartite hypergraph]
  An $r$-uniform bipartite hypergraph $H = (A, B, E)$ is a hypergraph
  on a vertex set partitioned into two sets $A$ and $B$ such that for
  every edge $e \in E$, $|e \cap A| = 1$ and $|e \cap B| = r-1$.
\end{definition}

Let $H=(A,B,E)$ be a $r$-uniform bipartite hypergraph. It is important
to note that we assume that the underlying bipartition of the vertex set is given
to the algorithm. We will
use $n$ and $m$ to refer to $|A|$ and $|E|$ respectively in $H$. A
subset of edges $M \subseteq E$ is called a \textit{partial matching}
if any pair of edges in the set are disjoint. A partial matching whose
edges contain every vertex of $A$ is a \textit{perfect matching}.

We need some notation for referring to the collection of $A$ vertices
and $B$ vertices in a set of edges $F \subseteq E$. For a subset of
edges $F \subseteq E$ we use $A(F)$ to denote the set
$\cup_{e\in F} \; e \cap A.$ $B(F)$ is defined similarly as
$\cup_{e\in F} \; e \cap B$.

We say that a vertex $a \in A$ is \textit{matched} by a partial
matching $M$ if $a \in A(M)$. Recall that a perfect matching is a
partial matching that matches all the vertices of $A$.

For the definitions that follow consider a fixed partial matching
$M$ in $H$. From the context it will always be clear what the
considered partial matching is.

\begin{definition}[Blocking edges]\label{def:blocking}
The set of
edges \textit{blocking} a given edge $e\in E$ is the set
\[\{f \in M \; | \; f \cap e \cap B \not = \emptyset\},\]
i.e., it contains edges in $M$ that prevent us from adding
$e$ to it.
\end{definition}
Note that $M$ may contain an edge $e'$ such that
$e' \cap e \cap B = \emptyset$ and it matches $\bar{a}$ in $M$, where $\{\bar{a}\} = e \cap A$, in
which case we may want to also add $e'$ to the set of blocking edges
of $e$ but we do \emph{not} do so according to Definition~\ref{def:blocking}.

An edge $e \in E$ is called \textit{immediately addable} if it has no
blocking edges. The name reflects the property that $M \cup \{e\}$ is
also a partial matching for such an edge $e$, unless the $A$ vertex
contained in $e$ is already matched by $M$. We refer to an edge $e \in E$ as an edge \emph{for}
$a \in A$ if $a \in e$.

The definitions are made with the following simple operation in mind.

\paragraph{Swapping operation} Suppose that $a$ is matched by
$M$ through some edge $e \in M$ and that there is
an immediately addable edge $f \in E$ for $a$. Then the set
$M \setminus \{e\} \cup \{f\}$ is also a partial matching
that matches exactly the same set of $A$ vertices as
$M$.

A final piece of notation is the following. For a collection of
indexed sets $\{S_0,S_1,\dots,S_k\}$ we write $S_{\leq t}$ to denote
\[\bigcup_{i=0}^t S_i.\]

\begin{definition}[Layer]\label{def:layer}
  A layer $L$ for a bipartite hypergraph $H= (A,B,E)$ with respect to
  a partial matching $M$ is a tuple $(X, Y)$ where
  \begin{itemize}
  \item $X \subseteq E \setminus M$,
  \item for each pair of distinct edges   $e, e' \in X$, $e \cap e'
    \cap B = \emptyset$,
  \item $Y \subseteq M$ is precisely the set of blocking edges of $X$, and
  \item every $e \in Y$ intersects exactly one edge from $X$.
  \end{itemize}
\end{definition}

\begin{definition}[Alternating tree]\label{def:alttree}
  An alternating tree $T$ for a bipartite hypergraph $H = (A, B, E)$
  with respect to a partial matching $M$ is a tuple 
  $(L_0,\dots,L_\ell)$ such that:
  \begin{itemize}
  \item $L_0 =(X_0, Y_0)$ is defined to be $(\emptyset, \{a_0\})$ for some 
  $a_0$ not matched by $M$,
  \item $L_1,\dots,L_\ell$ are layers,
\item $A(Y_{i-1}) \supseteq  A(X_i)$ for all $i=1,\dots,\ell$, and
\item $B(X_i \cup Y_i) \cap B(X_{i'}\cup Y_{i'}) = \emptyset \;
  \forall i \neq i' \in \{0,\dots,\ell\}$.
  \end{itemize}
  $a_0$ is called the \emph{root} of the alternating tree $T$. The
  \emph{degree} of an $A$ vertex $\bar{a} \in A(Y_{\leq \ell})$ is
  defined to be the number of edges from $(X_{\leq \ell} \cup Y_{\leq
    \ell}) \setminus Y_0$ that contain $\bar{a}$.
\end{definition}

\paragraph{Intuition} Our goal will be to obtain an augmenting
algorithm that takes some partial matching $M$ that does not match
some $a_0 \in A$ and turns it into a different partial matching $M'$
that matches all the vertices of $A(M)\cup \{a_0\}$. To accomplish
this consider some edge $e \in E$ for $a_0$. If it is immediately
addable then we are done. Otherwise there are some blocking edges of
$e$, call them $F$, that prevent us from adding $e$ to $M$. To make
progress we will try to perform a swapping operation on some of the
vertices from $A(F)$ thereby reducing the number of blocking edges of
$e$. To do so we need to find edges for $A(F)$ which may themselves
turn out to be blocked and so on. This alternating structure is
captured in our definition of a layer and the tree structure that
follows is the reason behind Definition~\ref{def:alttree}. See
Figure~\ref{figure:alttree} for an example of an alternating tree.

\begin{figure}
  \captionsetup{width=0.90\textwidth}
  \centering
  \includegraphics[width=0.75\textwidth, trim=0 3cm 0 1cm]{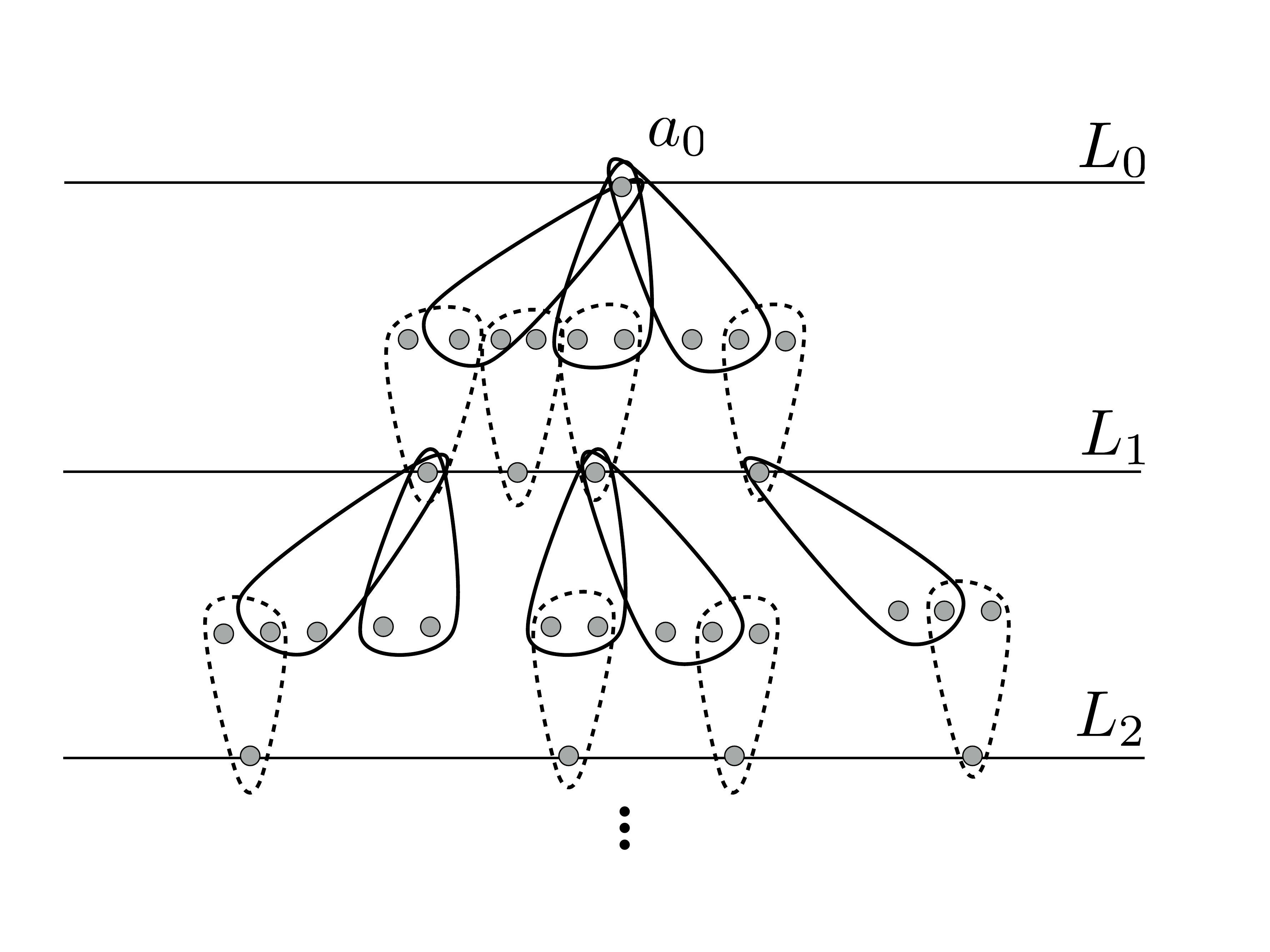}
  \caption{Alternating tree arising from some $3$-uniform bipartite
    hypergraph and partial matching, depicted here with a root $a_0$ and layers $L_1$ and $L_2$. The
    edges in $X_1$ and $X_2$ appear with a solid border, whereas
    (blocking) edges in $Y_1$ and $Y_2$ appear with a dotted border. The second edge from
the left with a solid border in $L_2$ is an example of an immedidately addable edge.}
  \label{figure:alttree}
\end{figure}


\paragraph{Degree bound} Our augmenting algorithm
depends on a single parameter
$\mu:= \epsilon^2/(10r^2).$ We also define $U := \lceil 1/\mu \rceil.$ As
$U+1$ will turn out to be an upper bound on the degree of any $A$
vertex in the alternating tree maintained by the augmenting algorithm,
we refer to $U$ as the degree bound. Note that every $A$ vertex in an
alternating tree $T$, except for the root, is part of exactly one
blocking edge, which follows from Definition~\ref{def:alttree} and the
fact that $M$ is a partial matching. Therefore, the degree bound implies that each non-root
$A$ vertex can be part of at most $U$ other (non-blocking) edges in
the alternating tree.

\begin{remark}
  Without loss of generality we will assume $0 < \epsilon < 1$. The
  parameters $\mu, U$ of our augmenting algorithm are set keeping in
  mind this range of values that $\epsilon$ can assume. If we knew
  stronger guarantees about the hypergraph $H$, for example, for
  $\epsilon$ as large as $10r$, then the values of these parameters
  can be set less aggressively and the running time bounds we obtain
  in later sections can also be improved drastically. Our goal here,
  however, is to show the existence of polynomial time algorithms even
  for a tiny advantage $\epsilon$.
\end{remark}

\section{The Augmenting Algorithm}\label{section:augmenting}
\subsection{The BuildLayer Subroutine}\label{section:augmenting1}
We first describe a subroutine $\textbf{BuildLayer}$ that is used by the
augmenting algorithm. It takes as input an alternating tree $T$, and a
pair of sets $X, Y \subseteq E$ that serve as the initial values for
the layer that the subroutine constructs. The subroutine augments $X$
and $Y$ and returns them at the end.

\paragraph{$\textbf{BuildLayer}(T, X, Y):$}
\begin{enumerate}[(a)]
\item We now describe what we mean by an ``addable edge'' for some
  given $X, Y \subseteq E$ and alternating tree $T$. Suppose $T = (L_0, \dots, L_\ell)$. For an $A$ vertex
  $\bar{a} \in A(Y_{\ell})$ we say that $\bar{a}$ has an \emph{addable
    edge} if i) $\bar{a}$ has fewer than $U$ edges in $X$, and ii)
  $\exists$ edge $e \in E$ for $\bar{a}$ disjoint from
  $B(X_{\leq \ell}\cup Y_{\leq \ell} \cup X \cup Y)$.
\item
    \textbf{While} there is an $\bar{a} \in A(Y_{\ell})$ having an
    addable edge $e \in E$, add $e$ to $X$ and its blocking
    edges to $Y$ as follows:
    \begin{eqnarray*}
      X &\leftarrow& X \cup \{ e \}, \\
      Y &\leftarrow& Y \cup \{f \in M \; | \; f \cap e
                              \cap B \neq \emptyset \}.
    \end{eqnarray*}
    \textbf{EndWhile}.
\item Return $(X, Y)$
\end{enumerate}

\subsection{Main Algorithm}\label{section:augmenting2}

We now describe the augmenting algorithm. The input to the algorithm
is a partial matching $M$ along with an $A$ vertex $a_0$ that is not
matched by $M$.

\paragraph{Initialization} Initialize layer $L_0$ in an alternating
tree $T$ by setting $(X_0, Y_0) \leftarrow (\emptyset, \{a_0\})$. The
variable $\ell$ will be updated to always point to the last layer in
the tree $T$. Set it to $0$.

\paragraph{Main Loop} Repeat the following two phases in order until
$a_0$ is matched by $M$.

\begin{enumerate}[(I)]
\item \textbf{Building phase}

  \begin{enumerate}
  \item Set $(X_{\ell+1}, Y_{\ell+1}) \leftarrow (\emptyset,
    \emptyset)$.
  \item $(X_{\ell+1}, Y_{\ell+1}) \gets \textbf{BuildLayer}(T,
    X_{\ell+1}, Y_{\ell+1}).$
  \item Add the new layer $L_{\ell+1}:=(X_{\ell+1}, Y_{\ell+1})$ to
    $T$.
  \item Increment $\ell$ to $\ell+1$.
  \end{enumerate}

\item \textbf{Collapse phase} Recall that $e \in E$ is
  \textit{immediately addable} if no edges from $M$ are blocking it,
  i.e., $f \cap e \cap B = \emptyset \; \forall f \in
  M$.

    \textbf{While} $X_{\ell}$ contains more than $\mu|X_{\ell}|$
    immediately addable edges, perform the following steps:

    For convenience, we call this set of steps in this iteration, the
    \emph{collapse} operation of layer $L_{\ell}$.

    \begin{enumerate}
    \item\label{step:lazy1} For each $f \in Y_{\ell-1}$ such that there is an immediately
      addable edge $e \in X_{\ell}$ for $\bar{a} \in A \cap f$,
      \begin{eqnarray*} 
        M &\leftarrow& M \setminus \{f\} \cup \{e\}, \\
        Y_{\ell-1} &\leftarrow& Y_{\ell-1} \setminus \{f\}. 
      \end{eqnarray*}
    \item Discard layer $L_\ell$ from $T$.
    \item\label{step:simulate} In this step we perform a
      \emph{superposed-build} operation on layer
      $L_{\ell-1} = (X_{\ell-1}, Y_{\ell-1})$ in $T$. Note that this
      layer is modified in this step iff the condition in
      Step~\ref{step:cond} is satisfied.
      \begin{enumerate}
      \item
        $(X'_{\ell-1}, Y'_{\ell-1}) \gets \textbf{BuildLayer}(T,
        X_{\ell-1}, Y_{\ell-1}).$
      \item\label{step:cond} If
        $|X'_{\ell-1}| \geq (1 + \mu)|X_{\ell-1}|$ then,
        $(X_{\ell-1}, Y_{\ell-1}) \gets (X'_{\ell-1}, Y'_{\ell-1})$
      \end{enumerate}
    \item $\ell \leftarrow \ell-1$.
    \end{enumerate}
    \textbf{EndWhile}.
\end{enumerate}

After the initialization, the main loop of the algorithm consists of
repeating the build and collapse phases in order. The state of the
algorithm at any moment is described by the alternating tree
$T=(L_0,\dots,L_\ell)$ and the partial matching $M$ maintained by the
algorithm, both of which are dynamically modified. It is not difficult
to verify that the addition of an extra layer $L_{\ell+1}$ in the
build phase and the collapse operations in the collapse phase modify
$T$ and $M$ in legal ways so that the resulting objects are consistent
with the definitions of an alternating tree and a partial matching,
respectively. We use these
facts without mention in the rest of the analayis.

Also note that set of vertices matched by $M$ always remains the same
throughout the execution of the algorithm until a collapse operation
on layer $L_1$ is performed, after which $M$ additionally matches
$a_0$, and the algorithm terminates.



\section{Analysis}\label{section:analysis}
We call a layer $L_i = (X_i, Y_i)$ \emph{collapsible} if more than $\mu|X_i|$ many edges in $X_i$ are immediately addable with respect to $M$. This is precisely the condition of the while loop in the collapse phase of the augmenting algorithm from Section~\ref{section:augmenting2}.


\begin{proposition}\label{lem:notcollapsible}
  Suppose that the alternating tree $T=(L_0,\dots,L_\ell)$ and the
  partial matching $M$ describe the state at the beginning of some
  iteration of the main loop of the augmenting algorithm. Then none of
  the layers $L_0,\dots,L_\ell$ are collapsible. As a corollary it
  follows that $|Y_{i}| \geq (1-\mu)|X_i|$ for each $i=1,\dots,\ell$.
\end{proposition}

\begin{proof}
  Suppose that the statement is true at the beginning of the current
  iteration. During the build phase a new layer $L_{\ell+1}$ is
  constructed. If $L_{\ell+1}$ is not collapsible then the claim
  follows for the beginning of the next iteration since none of the
  previous layers were modified in the current iteration. If
  $L_{\ell+1}$ turns out to be collapsible, then by the definition of
  the collapse phase the layers $L_0,\dots,L_t$ are left at the end of
  the collapse phase for some $t \leq \ell$ (note that $t \geq 1$
  unless $a_0$ was matched and the algorithm terminates in the current
  iteration). The state of each of the layers $L_0,\dots,L_{t-1}$ is
  unchanged from the beginning of the current iteration. Layer $L_t$
  on the other hand could have possibly been modified in
  Step~\ref{step:simulate} of the collapse phase. However, since it
  remains part of the alternating tree after the collapse phase it
  implies that $L_t$, subsequent to any modifications, is not
  collapsible. Therefore none of the layers in the alternating tree
  are collapsible at the end of the iteration (unless the algorithm
  terminates after the current iteration). Since the claim is true for
  the first iteration, the claim follows by induction on the number of
  iterations of the main loop of the augmenting algorithm.

  The corollary follows since $L_i$ is a layer, for each
  $i=1,\dots,\ell$, and, by Definition~\ref{def:layer}, $Y_i$ contains
  all the blocking edges of edges in $X_i$ and each edge in $Y_i$
  intersects (at most) one edge of $X_i$.
\end{proof}

Before we state the next proposition some clarification is necessary
concerning the description of the augmenting algorithm in
Section~\ref{section:augmenting}. For instance, in the building phase,
there could be many vertices $a \in A(Y_\ell)$ that have an
\emph{addable edge} (as defined in the Section~\ref{section:augmenting1}), and
even a a given vertex could take many addable edges from which one is
eventually chosen. The final state of layer $L_{\ell+1}$, at the
conclusion of the build phase, depends on the sum total of such
choices. The situation is similar in the collapse phase as well. In
order to properly specify the algorithm and refer to the quantities
maintained by it without ambiguity, we assume that there is a total
ordering on the vertices in $A \cup B$ and edges in $E$, and that
these orderings are used to choose a unique vertex and edge in any
event that many are admissible according to the algorithm description
in Section~\ref{section:augmenting}. This allows us, for example, to
refer precisely to \emph{the} layer $L_{\ell + 1}$ after performing a
build operation, or to \emph{the} layer $L'_\ell$ after performing a
superposed-build operation on layer $L_\ell$, etc.

\begin{proposition}\label{lem:noaugment}
  Suppose that the alternating tree $T=(L_0,\dots,L_\ell)$ and the
  partial matching $M$ describe the state at the beginning of some
  iteration of the main loop of the augmenting algorithm.  Then the
  superposed-build operation on $L_i$ (while ignoring layers
  $L_{i+1},\dots, L_\ell$)
  \[(X'_i, Y'_i) \gets \textbf{BuildLayer}((L_0,\dots, L_i), X_i,
    Y_i),\]
  where $L_i = (X_i, Y_i)$ satisfies $|X'_i| < (1+\mu)|X_i|$ for each
  $i=1,\dots, \ell$.
\end{proposition}

\begin{proof}
  Consider some layer $L_t$ for $0 \leq t \leq \ell$ present in the
  alternating tree at the beginning of the current iteration. At the
  iteration when layer $L_t$ was built a superposed-build operation
  could not have increased the size of $X_t$ even by one. If no
  collapse operations of some layer occurred until the current
  iteration then the situation remains identical, because layer
  $L_{t+1}$ was not collapsed in particular. If however, some layer
  was collapsed then it must have an index strictly greater than $t$
  (since otherwise, the algorithm would have discarded layer $L_t$ in
  that case). As every time layer $L_{t+1}$ is collapsed, and some
  edges from $Y_t$ are removed, the algorithm tries to augment $X_t$
  by a $\mu$ fraction when possible (in Step~\ref{step:simulate} of
  the collapse phase), it follows that the number of edges in $X_t$
  cannot increase by more than a $\mu$ fraction on performing
  superposed-build operation on $L_t$.
\end{proof}


To ensure that the algorithm does not get stuck we need to show that, for some state $(L_0,\dots,L_\ell)$ and $M$ reached at the beginning of an iteration of the main loop, the build phase creates a new layer $L_{\ell+1}$ with at least one edge. We prove the following stronger statement.

\begin{theorem}\label{lem:addable}
  Suppose $T=(L_0,\dots,L_\ell)$ is the alternating tree at the beginning of some iteration of the main loop of the augmenting algorithm and let $L_{\ell+1}$ be the newly constructed layer in the build phase of the iteration. Then,
  \[ |X_{i+1}| > \frac{\epsilon}{5r^2} |Y_{\leq i}|,\]
  for each $i=0,\dots,\ell$.
\end{theorem}

The proof of Theorem~\ref{lem:addable} uses Lemma~\ref{lem:addable0} and Lemma~\ref{lem:addable1}, which we prove below.

\begin{lemma}\label{lem:addable0}
  Suppose that $|Y_{\leq \ell}| < \lceil 5r^2/\epsilon \rceil$ at the beginning of some iteration of the main loop. Then when $L_{\ell+1}$ is built in the build phase of the iteration, $|X_{\ell+1}| \geq 1$.
\end{lemma}
\begin{proof}
  By Proposition~\ref{lem:notcollapsible},
  $|X_{\leq \ell}| \leq |Y_{\leq \ell}|/(1-\mu) < 6 r^2/\epsilon,$
  where we use that $\mu < 1/10$. By the choice of
  $\mu = \epsilon^2/(10r^2)$ we then have $\mu|X_{\leq \ell}| <
  1$. Therefore, by the invariants from
  Proposition~\ref{lem:notcollapsible} and
  Proposition~\ref{lem:noaugment}, every edge in $X_{\leq \ell}$ has
  at least one blocking edge in the tree, and no $A$ vertex with less
  than $U$ edges in $X_{\leq \ell}$ has an edge that is disjoint from
  $B(X_{\leq \ell} \cup Y_{\leq \ell})$. Next, no $A$ vertex in the
  tree can have $U = \lceil 1/\mu \rceil$ edges in the tree, as in
  that case the number of blocking edges for that vertex would be at
  least $U$ which is greater than $\lceil 5r^2/\epsilon \rceil$
  contradicting our hypothesis. Taking $S$ to be the set of all the
  $A$ vertices in the tree, so that $|S| = |Y_{\leq \ell}|$, these
  arguments show that the number of $B$ vertices in the tree is an
  upper bound on $\tau(E_S)$.

  We now show that the number of $B$ vertices in the tree is at most
  $(2r-3)(|S|-1)$. To see this, note that every non-root $A$ vertex in
  the tree is included in a unique edge in $M$ (also in the tree),
  which in turn intersects some unique edge in $X_{\leq
    \ell}$. Therefore for each non-root $A$ vertex in the tree, we
  have $(r-1)$ corresponding $B$ vertices from the matching edge in
  $M$ and an additional set of at most $(r-1)-1$ many $B$ vertices from the
  unique edge in the tree that intersects this matching edge. Further,
  this accounts for all the $B$ vertices in the tree. We may over
  count some $B$ vertices in edges that were added as addable edges in
  the alternating tree but this is fine since we are only aiming for
  an upper bound. So each non-root $A$ vertex can be throught to
  contribute at most $(r-1) + (r-2)$ many $B$ vertices to the tree.

  However, the guarantee is that $\tau(E_S)$ must be larger than
  $(2r-3+\epsilon)(|S|-1)$, which is a contradiction.
\end{proof}

We can say something stronger than Lemma~\ref{lem:addable0} when the
number of edges from $M$ in the alternating tree becomes $\Omega(r^2/\epsilon)$. Notice that this condition is satisfied when the number of layers in the alternating tree is $\Omega(r^2/\epsilon)$.

\begin{lemma}\label{lem:addable1}
  If $|Y_{\leq \ell}| \geq \lceil 5r^2/\epsilon \rceil$, at the beginning of some iteration then when layer $L_{\ell+1}$ is built
\[ |X_{\ell+1}| > \frac{\epsilon}{5r^2} |Y_{\leq \ell}|.\]
\end{lemma}
\begin{proof}
  Suppose that after the build phase constructing layer $L_{\ell+1}$ is complete, $|X_{\ell+1}| \leq \delta|Y_{\leq \ell}|$ where $\delta := \epsilon/(5r^2)$. 

  Let $S_i$ be the set of $A$ vertices from $A(Y_{i-1})$ that would
  take an addable edge if we were to perform a superposed-build
  operation on layer $L_i$ while ignoring layers
  $L_{i+1},\dots,L_{\ell+1}$. Formally, $S_i := A(X'_i \setminus
  X_i),$ where
  \[(X'_i, Y'_i) \gets \textbf{BuildLayer}((L_0,\dots, L_i), X_i,
    Y_i).\]

  Now define $S$ algorithmically (in the sense of performing steps in
  order) as follows:
  \begin{itemize}
  \item set $S$ to be the set of all $A$ vertices in $Y_{\leq \ell}$,
  \item remove all $A$ vertices from $S$ that have $U$ edges in the alternating tree (i.e., appear $U$ times in the edges in $X_{\leq \ell+1}$),
  \item remove all $A$ vertices in $\cup_{i=1}^{\ell} S_i$ from $S$.
  \end{itemize}
The number of $A$ vertices that have $U$ edges in the alternating tree is at most $|X_{\leq \ell+1}|/U$. The number of $A$ vertices in $\cup_{i=1}^\ell S_i$ is upper bounded by $\mu|X_{\leq \ell}|$ using Proposition~\ref{lem:noaugment}. Therefore, \[|S| \geq |Y_{\leq \ell}| - |X_{\leq \ell+1}|/U - \mu|X_{\leq \ell}|.\]
By our hypothesis towards contradiction $|X_{\ell+1}| \leq \delta|Y_{\leq \ell}|.$ Also, by Proposition~\ref{lem:notcollapsible}, $|X_{\leq \ell}| \leq |Y_{\leq \ell}|/(1-\mu).$ Putting these together,
\[ |S| \geq \left[1 - \left(\delta + \frac{1}{1-\mu}\right)\frac1U -\frac{\mu}{1-\mu} \right] |Y_{\leq \ell}|.\]

As $\tau(E_S) \geq (2r-3 + \epsilon)(|S|-1)$, we have 
\begin{align*}
 \tau(E_S) \geq \; (2r-3+\epsilon) \left(\left[1 - \left(\delta + \frac{1}{1-\mu}\right)\frac1U -\frac{\mu}{1-\mu} \right] |Y_{\leq \ell}|-1\right).
\end{align*}
Recall that $\mu = \epsilon^2/(10r^2) < 1/10$. So, after upper bounding the inner sum by
\begin{align*}
 \left(\delta + \frac{1}{1-\mu}\right)\frac1U +\frac{\mu}{1-\mu}  \leq \delta\cdot \frac1U + 2\frac{\mu}{1-\mu} < \frac{\epsilon}{5r^2}\cdot\frac{\epsilon^2}{10r^2} + 2\cdot\frac{\epsilon^2}{10r^2}\cdot\frac{10}{9} < \frac{\epsilon^2}{2r^2},
\end{align*}
we have 
\begin{equation}\label{eqn:largeHS2}
  \tau(E_S) > (2r-3+\epsilon)\left(\left[1 - \frac{\epsilon^2}{2r^2}\right] |Y_{\leq \ell}|-1\right).
\end{equation}
Next we obtain an upper bound on $\tau(E_S)$. We start by proving the following claim.

\begin{claim}\label{claim:local}
  $|B(X_{\leq \ell+1} \cup Y_{\leq \ell+1})| + \mu |X_{\leq
    \ell}|(r-1)^2$ is an upper bound on the cardinality of the
  smallest size hitting set for $E_S$ that is also a subset of $B$, i.e., an upper bound on $\tau(E_S)$.
\end{claim}
\begin{proof}
  From the definition of $S$, every vertex $a \in S$ appears in one of
  the layers $L_0,\dots,L_\ell$ and has strictly less than $U$ edges
  in the tree. Further, since each $a \in S$ is not part of
  $\cup_{i=1}^\ell S_i$ this means that there is no edge in $H$ for
  the vertex $a$ that is disjoint from the $B$ vertices in the tree
  and the $B$ vertices introduced in the superposed-build operations
  in each of the layers $L_1,\dots,L_\ell$. We now bound the total
  number of such $B$ vertices, to prove the claim.

The number of $B$ vertices present in the alternating tree
$T=(L_0,\dots,L_{\ell+1})$ is simply $|B(X_{\leq \ell+1} \cup Y_{\leq
  \ell+1})|$. Next, we know by Proposition~\ref{lem:noaugment} that a
superposed-build operation on a layer $L_i$ for $1\leq i\leq \ell$
produces a layer $L'_i$ such that $|X'_i| < (1+\mu) |X_i|$. Further,
the set of $B$ vertices introduced in $X'_i$ (as part of an addable
edge and their associated blocking edges) not already present in layer
$L_i$ (which was counted previously), is at most
$\mu|X_i|(r-1)^2$---each addable edge along with their blocking edges
contains at most $(r-1)^2$ many $B$ vertices.
\end{proof}

We now bound the total number of $B$ vertices in layers
$L_0,\dots,L_{\ell+1}$ in the alternating tree. The contribution from
layer $L_{\ell+1}$ is at most $|X_{\ell+1}|(r-1)^2$ since each addable
edge and its associated set of blocking edges can introduce at most
$(r-1)^2$ many $B$ vertices. Next, every edge in $X_{\leq \ell}$ is
either immediately addable or not. The $B$ vertices in immediately
addable edges from layers $L_0,\dots,L_\ell$ is at most $\mu|X_{\leq
  \ell}|(r-1)$ using Proposition~\ref{lem:notcollapsible}. The $B$
vertices from layers $L_0,\dots,L_\ell$ that are \emph{not} present in
immediately addable edges can be upper bounded simply by $(|Y_{\leq
  \ell}|-1)(2r-3)$ using the same argument as in Lemma~\ref{lem:addable0}.

Therefore, from Claim~\ref{claim:local} the following upper bound is then obtained for $\tau(E_s)$:
\begin{align*}
 (|Y_{\leq \ell}|-1)(2r-3) + \mu|X_{\leq \ell}|(r-1) + |X_{\ell+1}|(r-1)^2 + \mu |X_{\leq \ell}| (r-1)^2.
\end{align*}

We now explain the terms in the bound. The first three terms bound the number of $B$ vertices in the alternating tree as we saw above. The final term upper bounds contributions from edges not present in the alternating tree but those that could be added during the superposed-build operations on each of the layers $L_1,\dots,L_\ell$. Using the known bounds on $X_{\leq \ell} \leq Y_{\leq \ell}/(1-\mu)$ (from Proposition~\ref{lem:notcollapsible}) and $X_{\ell+1} \leq \delta |Y_{\leq \ell}|$ (by hypothesis),

\begin{align*}
\tau(E_S) < |Y_{\leq \ell}| \left[ (2r-3) + \frac{\mu}{1-\mu} \right. (r-1) + \delta(r-1)^2 + \left. \frac{\mu}{1-\mu}(r-1)^2 \right].
\end{align*}

For the chosen parameters $\mu = \epsilon^2/(10r^2)$ and $\delta=\epsilon/(5r^2)$, we get,

\begin{equation}\label{eqn:smallHS}
\tau(E_S) < \left[ 2r-3 + \epsilon/2\right]|Y_{\leq \ell}|.
\end{equation}

From \eqref{eqn:largeHS2} and \eqref{eqn:smallHS}, we have a contradiction when,

\[ \left[2r-3 + \epsilon/2 \right]|Y_{\leq \ell}| <  (2r-3+\epsilon)\left(\left[1 - \frac{\epsilon^2}{2r^2}\right] |Y_{\leq \ell}|-1\right),\]

which is true for $|Y_{\leq \ell}| \geq \lceil 5r^2/\epsilon \rceil$.
\end{proof}

We now complete the proof of Theorem~\ref{lem:addable}.

\begin{proof}[Proof of Theorem~\ref{lem:addable}]
  First notice that after a layer $L_{i+1}$ is built, the number of
  edges in $X_{i+1}$ is non-decreasing until it is collapsed in some
  future iteration. Also, any collapse operation on a layer leads to
  discarding that layer. Then the claim follows by combining
  Lemma~\ref{lem:addable0} and Lemma~\ref{lem:addable1} to note that
  at the moment when layer $L_{i+1}$ is created, for some $0 \leq i
  \leq \ell$, the inequality $|X_{i+1}| > \epsilon/(5r^2) \cdot |Y_{\leq i}|$ holds. This also remains true in future iterations until a collapse operation occurs in layer $L_{i+1}$, in which case it will no longer be part of the alternating tree maintained by the augmenting algorithm. 
\end{proof}

We are now in a position to bound the number of layers in the alternating tree at any point in the execution of the augmenting algorithm.

\begin{lemma}\label{lem:layers}
  The number of layers in the alternating tree $T=(L_0,\dots,L_\ell)$ maintained during the execution of the augmenting algorithm is always bounded by $O(\log n)$.
\end{lemma}

\begin{proof}
Suppose there are $\ell+1$ layers $L_0,\dots,L_\ell$ at the beginning of some iteration of the main loop of the augmenting algorithm. Consider some layer $L_i$ for $1 \leq i \leq \ell$. By Proposition~\ref{lem:notcollapsible} less than $\mu$ fraction of $X_{i}$ are immediately addable, and hence $|Y_{i}| > (1-\mu)|X_{i}|$. Then by Theorem~\ref{lem:addable} we have \[|Y_i| > \frac{(1-\mu)\epsilon}{(5r^2)} |Y_{\leq i-1}|\] for each $i = 1,\dots, \ell$. This quickly yields $n = |A| \geq |Y_{\leq \ell}| > (1 + (1-\mu)\epsilon/(5r^2))^\ell|Y_0|$, so that $n \geq (1 +\gamma)^\ell$, where $\gamma := (1-\mu)\epsilon/(5r^2)  > 0$. Altogether this implies that the number of layers $\ell$ at any moment in the algorithm is bounded by $O(\log n)$.
\end{proof}

\section{Signature Vectors}

To keep track of the progress made by the augmenting algorithm we design a potential function. For a given state of the alternating tree with layers $L_0,\dots,L_\ell$ in total we define the signature of layer $L_i$ (for $1 \leq i \leq \ell$) as:
\begin{equation}\label{eqn:signature}
\begin{split}
\text{signature}& \text{\;of \;} L_i := \left(-\lfloor \log_b \frac{(5r^2/\epsilon)^i }{(1-\mu)^{i-1}} |X_i| \rfloor, \lfloor \log_b \frac{(5r^2/\epsilon)^i}{(1-\mu)^i} |Y_i| \rfloor    \right),
\end{split}
\end{equation}
where $b := \frac{1}{1-\mu^3}$. The potential function associated with
the alternating tree at any state is the sequence obtained by
concatenating the signatures of the individual layers in order, and
finally appending the symbol $\infty$ at the end. We refer to this
potential function as the signature vector. In total there are
$2\ell+1$ coordinates in the signature vector which we write as
$s = (s_1,\dots,s_{2\ell},\infty)$.


\begin{lemma}\label{lem:decrease}
  The lexicographic value of the signature vector reduces across each iteration of the main loop in the augmenting algorithm unless the algorithm terminates during that iteration.
\end{lemma}
\begin{proof}
  Suppose the alternating tree $T=(L_0,\dots,L_\ell)$ and the partial matching $M$ define the state of the algorithm at the beginning of the iteration. Let the signature of the corresponding alternating tree be $s=(s_1,\dots,s_{2\ell},\infty)$. We consider two cases depending on whether a collapse operation occurred during the collapse phase of the current iteration.
  
  \begin{itemize}
  \item \textbf{No collapse operation occurred.} $\;$ In this case only the build phase of the iteration modified the state of the algorithm by adding a new layer $L_{\ell+1}$. Thus, the new signature of the alternating tree is $s'=(s'_1,\dots,s'_{2\ell+2},\infty)$ where $s'_i = s_i$ for all $i \leq 2\ell$ and $(s'_{2\ell+1}, s'_{2\ell+2})$ are defined as in \eqref{eqn:signature} for layer $L_{\ell+1}$ at the beginning of the next iteration. Clearly the lexicographic value of the signature of the alternating tree has reduced.

  \item \textbf{At least one collapse operation occurred.} $\;$ This means that during the iteration a new layer $L_{\ell+1}$ was built, and one or more collapse operations occurred in the collapse phase. Let primed quantities denote the variables after the end of the collapse phase in the iteration. Suppose that $t$ ($\leq \ell+1$) is the index of the earliest layer that was collapsed among all the collapse operations in the collapse phase in the iteration. If $t=1$ then $a_0$ was matched and the algorithm terminates. Otherwise $t > 1$ and by the description of the algorithm, the only layers left in the alternating tree after the collapse phase are $L'_0,\dots,L'_{t-1}$ where $L'_i$ is identical to $L_i$ for all $i < t-1$. Thus the new signature after the collapse phase is $s'=(s'_1,\dots,s'_{2t-2},\infty)$ where $s'_i = s_i$ for all $i \leq 2t-4$ and,
\begin{equation*}
\begin{split}
(s'_{2t-3}, s'_{2t-2})  = \bigg( -\lfloor \log_b \frac{(5r^2/\epsilon)^i }{(1-\mu)^{i-1}} |X'_{t-1}| \rfloor,\lfloor \log_b \frac{(5r^2/\epsilon)^i}{(1-\mu)^i} |Y'_{t-1}| \rfloor    \bigg).
\end{split}
\end{equation*}
When layer $L_t$ was collapsed Step~\ref{step:simulate} of the collapse phase could have possibly modified layer $L_{t-1}$. Accordingly there are two subcases.

\begin{itemize}
\item $|X'_{t-1}| = |X_{t-1}|.$ Since there was no modification to $X_{t-1}$ we look at how $Y_{t-1}$ has changed. As we collapsed layer $L_t$ in the alternating tree, there must have been at least $\mu |X_t|$ immediately addable edges in $L_t$. These must have caused the removal of at least $\mu|X_t|/U$ many matching edges in $L_{t-1}$. Further, by Theorem~\ref{lem:addable}, $|X_t| > \epsilon/(5r^2) |Y_{t-1}|$. Together this means that 
\begin{align*}
|Y'_{t-1}| < (1-(\mu/U) \cdot \epsilon/(5r^2))|Y_{t-1}| < (1-\mu^3)|Y_{t-1}|.
\end{align*} By our choice of the base of the logarithm it holds that $\log_b \frac{(5r^2/\epsilon)^i}{(1-\mu)^i} |Y'_i| \leq \log_b (1-\mu^3) + \log_b \frac{(5r^2/\epsilon)^i}{(1-\mu)^i} |Y_i| = -1 + \log_b \frac{(5r^2/\epsilon)^i}{(1-\mu)^i} |Y_i|$. Therefore,  $s'_{2t-3} = s_{2t-3}$ whereas $s'_{2t-2} < s_{2t-2}$.

\item $|X'_{t-1}| \geq (1+\mu)|X_{t-1}|.$ In this subcase the fact that $(1+\mu) \geq b$ implies that the lexicographic value of the signature vector has reduced since $s'_{2t-3} < s_{2t-3}$.
\end{itemize}
  \end{itemize}
\end{proof}

To show that the augmenting algorithm terminates in polynomial time we need one more fact.

\begin{proposition}\label{lem:nondecreasing}
  The coordinates of the signature vector are non-decreasing in absolute value at the beginning of each iteration of the main loop of the augmenting algorithm.
\end{proposition}
\begin{proof}
  Consider some layer $L_i$ for $1 \leq i \leq \ell$. Clearly the corresponding pair of coordinates $s_{2i-1},s_{2i}$ in the signature vector are non-decreasing in absolute value since $|Y_i| \geq (1-\mu)|X_i|$ using Proposition~\ref{lem:notcollapsible}. Between any two layers, by Theorem~\ref{lem:addable} we have $|X_i| > \epsilon/(5r^2) |Y_{i-1}|$ and so $|s_{2i-2}| \leq |s_{2i-1}|$. Thus the coordinates are non-decreasing in absolute value in the signature vector.
\end{proof}

\begin{lemma}\label{lem:sigs}
  The number of signature vectors is bounded by a polynomial in $n$.
\end{lemma}
\begin{proof}
  By Lemma~\ref{lem:layers} we know that the signature vector has at
  most $O(\log n)$ coordinates. By Proposition~\ref{lem:nondecreasing}
  the coordinates are also integers that are non-decreasing in
  absolute value. At this point one can obtain a trivial bound of
  $O(\log n)$ on the absolute value of each coordinate of the
  signature vector using the definition in
  \eqref{eqn:signature} and Lemma~\ref{lem:layers}. Since the sign
  pattern of the signature vector is always fixed, each signature
  vector can be thought to describe a unique partition of some
  positive integer of size at most $O(\log^2 n)$. Recall that a
  \emph{partition} of a positive integer $N$ is a way of writing $N$
  as the sum of positive integers without regard to order. Since the
  number of partitions of an integer of size $t$ is (asymptotically)
  at most $c^{\sqrt{t}}$ for some absolute constant
  $c > 1$~\cite{HardyRam18}, the claim then follows.
\end{proof}

As noticed by one of the reviewers, the dependence of
Lemma~\ref{lem:sigs} on the asymptotics of the partition function can
be avoided by modifying the signature vector to ensure that its
entries are strictly increasing in absolute value (instead of simply
being non-decreasing as in Lemma~\ref{lem:nondecreasing}), thereby
allowing a signature vector to be inferred by specifying a subset of a
set of size at most $O(\log n)$. One way to get this property is by
adding/subtracting $i$ to the $i$-th coordinate of the signature
vector, consistent with its sign pattern.

We are now in a position to use the potential function defined in this
section to wrap up the proof of our main result.

\begin{proof}[Proof of Theorem~\ref{thm:main}]
  From Lemma~\ref{lem:decrease} we have that every iteration of the main loop of the augmenting algorithm described in Section~\ref{section:augmenting} reduces the lexicographic value of the signature vector. Lemma~\ref{lem:sigs} further tells us that the number of such signature vectors is bounded by a polynomial in $n$. Thus, the augmenting algorithm terminates in polynomially many iterations. It can also be verified that each iteration of the augmenting algorithm can be implemented to run in time polynomial in $n$ and $m$. Finally, running the augmenting algorithm $n$ times, starting with an empty partial matching, yields the desired perfect matching in $H$.  \end{proof}


From the proof of Lemma~\ref{lem:addable1} we also note that the
algorithm in Section~\ref{section:augmenting} can be suitably modified
to yield the following
slightly stronger version of Theorem~\ref{thm:main} as a corollary.
\begin{theorem}\label{thm:maincorollary}
  For every fixed choice of $\epsilon >0$ and $r \geq 2$, there exists
  an algorithm $\mathcal{A}'(\epsilon, r)$ that takes as input an
  $r$-uniform bipartite hypergraph $H=(A,B,E)$, runs in polynomial
  time, and terminates after finding either:
  \begin{itemize}
  \item a perfect matching in $H$, or
  \item a set $S \subseteq A$ such that $\tau(E_S) \leq (2r-3+\epsilon)(|S|-1).$
  \end{itemize}
\end{theorem}

\section{Conclusion and Open Problems}
In this paper we presented a polynomial time algorithm for finding
perfect matchings in bipartite hypergraphs satisfying a slightly
stronger version of Haxell's condition. The algorithm is essentially
the natural generalization of the well known Hungarian algorithm for
finding perfect matchings in graphs with two essential modifications:
i) restricting the degree of vertices in the constructed alternating
tree, and ii) performing updates on the alternating tree lazily. The
two ideas in tandem give us a polynomial running time bound on the
procedure.

One subtlety here is that the algorithm performs lazy updates in two
places, in Steps~\ref{step:lazy1} and~\ref{step:simulate}, in the
collapse phase. While the former is crucial for the running time
bound, the latter seems to be an artifact of the analysis driven by
the specific choice of the signature vector in
Section~\ref{section:analysis}. In particular, this can likely be
avoided by choosing a different signature vector to measure progress.

Finally, we point out the obvious open problem in this line of work.

\begin{question}
  Does there exist such an algorithm with a polynomial running time
  dependence on at least one of the parameters $1/\epsilon$ and $r$?
\end{question}

\section*{Acknowledgements}
We thank Yuri Faenza and Ola Svensson for providing helpful comments
on an earlier draft of this paper. We also thank anonymous SODA
reviewers for their valuable comments that helped improved the
presentation.

\bibliographystyle{alpha}
\bibliography{refs}





\end{document}